\documentclass[11pt]{article}
\usepackage{fullpage}
\usepackage{amsmath}
\usepackage{amsthm}
\usepackage{amssymb}
\usepackage{latexsym}
\usepackage{enumerate}
\usepackage{palatino}
\usepackage{graphicx}
\usepackage{color}

\newcommand{\Xomit}[1]{}
\newtheorem{theorem}{Theorem}[section]
\newtheorem{lemma}[theorem]{Lemma}
\newtheorem{definition}{Definition}[section]

\newcommand{\expand}{\mbox{\sc expand}}
\newcommand{\round}{\mbox{\sc decr}}
\newcommand{\reduce}{\mbox{\sc reduce}}
\newcommand{\descend}{\mbox{\sc descend}}
\newcommand{\density}{\textit{density}}
\newcommand{\depth}{\textit{depth}}
\newcommand{\girth}{\textit{girth}}
\newcommand{\radius}{\textit{orad}}
\newcommand{\squares}{\textit{squares}}
\newcommand{\wmax}{\textit{wmax}}
\newcommand{\border}{\textit{Border}}
\newcommand{\inherit}{\textit{Inherit}}
\newcommand{\treeroot}{\textit{Root}}
\newcommand{\leaf}{\textit{Leaf}}
\newcommand{\nonleaf}{\textit{Nonleaf}}
\newcommand{\setabove}{\textit{Above}}
\newcommand{\setbelow}{\textit{Below}}
\newcommand{\lchild}{\textit{Lchild}}
\newcommand{\rchild}{\textit{Rchild}}
\newcommand{\oracle}{\textit{Oracle}}
\newcommand{\lookuptable}{\textit{Table}}

\title{Computing the Girth of a Planar Graph in Linear Time\thanks{A
    preliminary version appeared in COCOON~2011~\cite{ChangL11}. The current version is accepted to {\em SIAM Journal on Computing}.}}

\author{Hsien-Chih Chang\thanks{Email: {\tt{hchang17@illinois.edu}}.
    Department of Computer Science and Information Engineering,
    National Taiwan University.}
\and
Hsueh-I Lu\thanks{Corresponding author. Email:
  {\tt{hil@csie.ntu.edu.tw}}. Web: {\tt{www.csie.ntu.edu.tw/\~{}hil}}.
  Department of Computer Science and Information Engineering, National
  Taiwan University.  This author also holds joint appointments in the
  Graduate Institute of Networking and Multimedia and the Graduate
  Institute of Biomedical Electronics and Bioinformatics, National
  Taiwan University. Address: 1 Roosevelt Road, Section 4, Taipei 106,
  Taiwan, ROC. Research supported in part by NSC grants
  98--2221--E-002--079--MY3 and 101--2221--E--002--062--MY3.}  }

\date{April 22, 2013}

\begin{document}

\maketitle

\begin{abstract}
The {\em girth} of a graph is the minimum weight of all simple cycles
of the graph.  We study the problem of determining the girth of an
$n$-node unweighted undirected planar graph.  The first non-trivial
algorithm for the problem, given by Djidjev, runs in $O(n^{5/4}\log
n)$ time.  Chalermsook, Fakcharoenphol, and Nanongkai reduced the
running time to $O(n\log^2 n)$.  Weimann and Yuster further reduced
the running time to $O(n\log n)$.  In this paper, we solve the problem
in $O(n)$ time.
\end{abstract}

\section{Introduction}
Let $G$ be an edge-weighted simple graph, i.e., $G$ does not contain
multiple edges and self-loops.  We say that $G$ is {\em unweighted} if
the weight of each edge of $G$ is one.  A cycle of $G$ is {\em simple}
if each node and each edge of $G$ is traversed at most once in the
cycle.  The {\em girth} of $G$, denoted $\girth(G)$, is the minimum
weight of all simple cycles of $G$.  For instance, the girth of each
graph in Figure~\ref{figure:fig1} is four.  As shown by, e.g.,
Bollob\'as~\cite{Bollobas78}, Cook~\cite{Cook75}, Chandran and
Subramanian~\cite{ChandranS05}, Diestel~\cite{Diestel00},
Erd\H{o}s~\cite{Erdos59}, and Lov\'{a}sz~\cite{Lovasz68}, girth is a
fundamental combinatorial characteristic of graphs related to many
other graph properties, including degree, diameter, connectivity,
treewidth, and maximum genus.  We address the problem of computing the
girth of an $n$-node graph.  Itai and Rodeh~\cite{ItaiR78} gave the
best known algorithm for the problem, running in time $O(M(n) \log
n)$, where $M(n)$ is the time for multiplying two $n\times n$
matrices~\cite{CoppersmithW90}.  In the present paper, we focus on the
case that the input graph is undirected, unweighted, and planar.
Djidjev~\cite{Djidjev00,Djidjev10} gave the first non-trivial
algorithm for the case, running in $O(n^{5/4} \log n)$ time.  The
min-cut algorithm of Chalermsook, Fakcharoenphol, and
Nanongkai~\cite{ChalermsookFN04} reduced the time complexity to $O(n
\log^2 n)$, using the maximum-flow algorithms of, e.g., Borradaile and
Klein~\cite{BorradaileK09} or Erickson~\cite{Erickson10}.  Weimann and
Yuster~\cite{WeimannY10} further reduced the running time to $O(n\log
n)$.  Linear-time algorithms for an undirected unweighted planar graph
were known only when the girth of the input graph is bounded by a
constant, as shown by Itai and Rodeh~\cite{ItaiR78}, Vazirani and
Yannakakis~\cite{VaziraniY88}, and Eppstein~\cite{Eppstein99}.  We
give the first optimal algorithm for any undirected unweighted planar
graph.
\begin{theorem}
\label{theorem:theorem1.1}
The girth of an $n$-node undirected unweighted planar graph is
computable in $O(n)$ time.
\end{theorem}

\begin{figure}[t]
\centerline{\scalebox{0.85}{\input{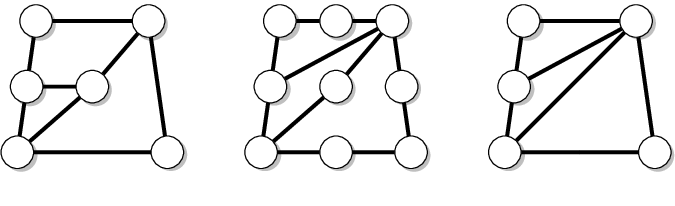}}}
\caption{(a) A planar graph $G$ with nonnegative integral edge
  weights. (b) The expanded version $\expand(G)$ of $G$.  (c) A
  contracted graph $G'$ with $\expand(G')=\expand(G)$.}
\label{figure:fig1}
\end{figure}

{\em Related work}.  The $O(M(n) \log n)$-time algorithm of Itai and
Rodeh~\cite{ItaiR78} also works for directed graphs.  The best known
algorithm for directed planar graphs, due to Weimann and
Yuster~\cite{WeimannY10}, runs in $O(n^{3/2})$ time.  The $O(n\log^2
n)$-time algorithm of Chalermsook et al.~\cite{ChalermsookFN04}, using
the maximum-flow algorithms of Borradaile and
Klein~\cite{BorradaileK09} or Erickson~\cite{Erickson10} also works
for undirected planar graphs with nonnegative weights.  The recent
max-flow algorithm of Italiano, Nussbaum, Sankowski, and
Wulff-Nilsen~\cite{ItalianoNSW11} improved the running time of
Chalermsook et al.~to $O(n\log n \log\log n)$.  For any given constant
$k$, Alon, Yuster, and Zwick~\cite{AlonYZ95} showed that a $k$-edge
cycle of any $n$-node general graph, if it exists, can be found in
$O(M(n) \log n)$ time or expected $O(M(n))$ time. The time complexity
was reduced to $O(n^2)$ by Yuster and Zwick~\cite{YusterZ97}
(respectively, $O(n)$ by Dorn~\cite{Dorn10}) if $k$ is even
(respectively, the input graph is planar).  See,
e.g.,~\cite{GomoryH61,Monien83,HaoO94,NagamochiI92,KargerS96,Karger00,LingasL09,ShihWK90,AlonYZ97,Mader98,DiestelR04,Hayes03,LinialMN02,Molloy02,DyerF01,Kochol05,KostochkaKSS10,KawarabayashiT09,Cabello10,CabelloVL10a,EricksonW10,WangL03,ItalianoNSW11}
for work related to girths and min-weight cycles in the literature.

{\em Overview.}  The {\em degree} of a graph is the maximum degree of
the nodes in the graph.  For instance, the number of neighbors of each
node in an $O(1)$-degree graph is bounded by an absolute constant.  To
compute $\girth(G_0)$ for the input $n$-node planar graph $G_0$, we
turn $G_0$ into an $m$-node ``contracted''
(see~\S\ref{subsection:contract}) graph $G'$ with positive integral
edge weights such that $m\leq n$ and $\girth(G')=\girth(G_0)$, as done
by Weimann and Yuster~\cite{WeimannY10}.  If the ``density''
(see~\S\ref{subsection:contract}) of $G'$ is $\Omega(\log^2 m)$, we
can afford to use the algorithm of Chalermsook et
al.~\cite{ChalermsookFN04} (see Theorem~\ref{lemma:lemma2.1}) to
compute $\girth(G')$.  Otherwise, by $\girth(G')=O(\log^2 m)$, as
proved by Weimann and Yuster (see Lemma~\ref{lemma:lemma2.4}), and the
fact $G'$ has positive integral weights, we can further transform $G'$
to a $\Theta(m)$-node $O(\log^2 m)$-outerplane graph $G$ with $O(1)$
degree, $O(\log^2 m)$ density, and $O(\log^2 m)$ maximum weight such
that $\girth(G)=\girth(G')$.  The way we reduce the ``outerplane
radius'' (see~\S\ref{subsection:radius}) is similar to those of
Djidjev~\cite{Djidjev10} and Weimann and Yuster~\cite{WeimannY10}.  In
order not to increase the outerplane radius, our degree-reduction
operation (see \S\ref{subsection:radius}) is different from that of
Djidjev~\cite{Djidjev10}.  Although $G$ may have zero-weight edges and
may no longer be contracted, it does not affect the correctness of the
following approach for computing $\girth(G)$.

A cycle of a graph is {\em non-degenerate} if some edge of the graph
is traversed exactly once in the cycle.  Let $u$ and $v$ be two
distinct nodes of $G$.  Let $g(u,v)$ be the minimum weight of any
simple cycle of $G$ that contains $u$ and $v$.  Let $d(u,v)$ be the
distance of $u$ and $v$ in $G$. For any edge $e$ of $G$, let
$d(u,v;e)$ be the distance of $u$ and $v$ in $G\setminus\{e\}$.  If
$e(u,v)$ is an edge of some min-weight path between $u$ and $v$ in
$G$, then $d(u,v)+d(u,v;e(u,v))$ is the minimum weight of any
non-degenerate cycle containing $u$ and $v$ that traverses $e(u,v)$
exactly once. In general, $d(u,v)+d(u,v;e(u,v))$ could be less than
$g(u,v)$. However, if $u$ and $v$ belong to a min-weight simple cycle
of $G$, then $d(u,v)+d(u,v;e(u,v))=g(u,v)=\girth(G)$.

Computing the minimum $d(u,v)+d(u,v;e(u,v))$ over all pairs of nodes
$u$ and $v$ in $G$ is too expensive. However, computing
$d(u,v)+d(u,v;e(u,v))$ for all pairs of nodes $u$ and $v$ in a small
node set $S$ of $G$ leads to a divide-and-conquer procedure for
computing $\girth(G)$.  Specifically, since $G$ is an $O(\log^2
m)$-outerplane graph, there is an $O(\log^2 m)$-node set $S$ of $G$
partitioning $V(G)\setminus S$ into two non-adjacent sets $V_1$ and
$V_2$ with roughly equal sizes.  Let~$C$ be a min-weight simple cycle
of $G$. Let $G_1$ (respectively, $G_2$) be the subgraph of $G$ induced
by $V_1\cup S$ (respectively, $V_2\cup S$).  If $V(C)\cap S$ has at
most one node, the weight of $C$ is the minimum of $\girth(G_1)$ and
$\girth(G_2)$.  Otherwise, the weight of $C$ is the minimum
$d(u,v)+d(u,v;e(u,v))$ over all $O(\log^4 m)$ pairs of nodes $u$ and
$v$ in $S$.  Edges $e(u,v)$ and distances $d(u,v)$ and $d(u,v;e(u,v))$
in $G$ can be obtained via dynamic programming from edges $e(u,v)$ and
distances $d(u,v)$ and $d(u,v;e(u,v))$ in $G_1$ and $G_2$ for any two
nodes $u$ and $v$ in an $O(\log^3 m)$-node superset ``$\border(S)$''
(see~\S\ref{section:task1}) of $S$.  The above recursive procedure
(see Lemma~\ref{lemma:lemma5.4}) is executed for two levels.  The
first level (see the proofs of Lemmas~\ref{lemma:lemma3.1} and
\ref{lemma:lemma5.4}) reduces the girth problem of $G$ to girth and
distance problems of graphs with $O(\log^{30} m)$ nodes. The second
level (see the proofs of Lemmas~\ref{lemma:lemma5.6} and
\ref{lemma:lemma6.1}) further reduces the problems to girth and
distance problems of graphs with $O((\log\log m)^{30})$ nodes, each of
whose solutions can thus be obtained directly from an $O(m)$-time
pre-computable data structure (see Lemma~\ref{lemma:lemma5.5}).  Just
like Djidjev~\cite{Djidjev10} and Chalermsook et
al.~\cite{ChalermsookFN04}, we rely on dynamic data structures for
planar graphs.  Specifically, we use the dynamic data structure of
Klein~\cite{Klein05} (see Lemma~\ref{lemma:lemma5.2}) that supports
point-to-point distance queries.  We also use Goodrich's decomposition
tree~\cite{Goodrich95} (see Lemma~\ref{lemma:lemma4.2}), which is
based on the link-cut tree of Sleator and Tarjan~\cite{SleatorT83}.
The interplay among the densities, outerplane radii, and maximum
weights of subgraphs of $G$ is crucial to our analysis.  Although it
seems unlikely to complete these two levels of reductions in $O(m)$
time, we can fortunately bound the overall time complexity by $O(n)$.

The rest of the paper is organized as follows.
Section~\ref{section:prelim} gives the preliminaries and reduces the
girth problem on a general planar graph to the girth problem on a
graph with $O(1)$ degree and poly-logarithmic maximum weight,
outerplane radius, and density.  Section~\ref{section:framework} gives
the framework of our algorithm, which consists of three tasks.
Section~\ref{section:task1} shows Task~1.  Section~\ref{section:task2}
shows Task~2.  Section~\ref{section:task3} shows Task~3.
Section~\ref{section:conclude} concludes the paper.

\section{Preliminaries}
\label{section:prelim}

All logarithms throughout the paper are to the base of two.  Unless
clearly specified otherwise, all graphs are undirected simple planar
graphs with nonnegative integral edge weights.
Let $|S|$ denote the cardinality of set~$S$.  Let $V(G)$ consist of
the nodes of graph $G$.  Let $E(G)$ consist of the edges of graph
$G$. Let $|G|=|V(G)|+|E(G)|$.  By planarity of $G$, we have
$|G|=\Theta(|V(G)|)$.  Let $\wmax(G)$ denote the maximum edge weight
of $G$.  For instance, if $G$ is as shown in
Figures~\ref{figure:fig1}(a) and~\ref{figure:fig1}(b), then 
$\wmax(G)=2$ and $\wmax(G)=1$, respectively.  Let $w(G)$ denote the
sum of edge weights of graph $G$.  Therefore, $\girth(G)$ is the
minimum $w(C)$ over all simple cycles $C$ of $G$.

\begin{theorem}[Chalermsook et al.~\cite{ChalermsookFN04}]
\label{lemma:lemma2.1}
If $G$ is an $m$-node planar graph with nonnegative weights, then it
takes $O(m\log^2 m)$ time to compute $\girth(G)$.
\end{theorem}

\subsection{Expanded version, density, weight decreasing, and contracted graph}
\label{subsection:contract}
The {\em expanded version} of graph $G$, denoted $\expand(G)$, is the
unweighted graph obtained from $G$ by the following operations: (1)
For each edge $(u,v)$ with positive weight $k$, we replace edge
$(u,v)$ by an unweighted path $(u,u_1,u_2,\ldots,u_{k-1},v)$; and (2)
for each edge $(u,v)$ with zero weight, we delete edge $(u,v)$ and
merge $u$ and $v$ into a new node.  For instance, the graph in
Figure~\ref{figure:fig1}(b) is the expanded version of the graphs in
Figures~\ref{figure:fig1}(a) and~\ref{figure:fig1}(c).  One can verify
that the expanded version of $G$ has $w(G)-|E(G)|+|V(G)|$ nodes.
Define the {\em density} of $G$ to be
\begin{displaymath}
\density(G) = \frac{|V(\expand(G))|}{|V(G)|}.
\end{displaymath}
For instance, the densities of the graphs in
Figures~\ref{figure:fig1}(a) and~\ref{figure:fig1}(c)
are~$\frac{3}{2}$ and $\frac{9}{5}$, respectively.  
\begin{lemma}
\label{lemma:lemma2.2}
The following statements hold for any graph $G$.
\begin{enumerate}
\item 
\label{item:lemma2.2:item1}
$\girth(\expand(G))=\girth(G)$.

\item 
\label{item:lemma2.2:item2}
$\density(G)$ can be computed from $G$ in $O(|G|)$ time.
\end{enumerate}
\end{lemma}

For any number $w$, let $\round(G,w)$ be the graph obtainable in
$O(|G|)$ time from $G$ by decreasing the weight of each edge $e$ with
$w(e)>w$ down to $w$. The following lemma is straightforward.
\begin{lemma}
\label{lemma:lemma2.3}
If $G$ is a graph and $w$ is a positive integer, then
$\density(\round(G,w))\leq\density(G)$.  Moreover, if
$w\geq\girth(G)$, then $\girth(\round(G,w))=\girth(G)$.
\end{lemma}

A graph is {\em contracted} if the two neighbors of any degree-two
node of the graph are adjacent in the graph.  For instance, the graphs
in Figures~\ref{figure:fig1}(a) and~\ref{figure:fig1}(b) are not
contracted and the graph in Figure~\ref{figure:fig1}(c) is contracted.

\begin{lemma}
[{Weimann and Yuster~\cite[Lemma~3.3]{WeimannY10}}]\
\label{lemma:lemma2.4}
\begin{enumerate}
\item 
\label{item:lemma2.4:item1}
Let $G_0$ be an $n$-node unweighted biconnected planar graph.  It
takes $O(n)$ time to compute an $m$-node biconnected contracted planar
graph $G$ with positive integral weights such that $m\leq n$ and
$G_0=\expand(G)$.

\item
\label{item:lemma2.4:item2}
If $G$ is a biconnected contracted planar graph with positive integral
weights, then we have that $\girth(G)\leq 36\cdot
\density(G)$.
\end{enumerate}
\end{lemma}

\subsection{Outerplane radius and degree reduction}
\label{subsection:radius}

A {\em plane graph} is a planar graph equipped with a planar
embedding.  A node of a plane graph is {\em external} if it is on the
outer face of the embedding.  The {\em outerplane depth} of a node $v$
in a plane graph $G$, denoted $\depth_G(v)$, is the positive integer
such that $v$ becomes external after peeling $\depth_G(v)-1$ levels of
external nodes from $G$.  The {\em outerplane radius} of $G$, denoted
$\radius(G)$, is the maximum outerplane depth of any node in $G$.  A
plane graph $G$ is {\em $r$-outerplane} if $\radius(G)\leq r$.  For
instance, in the graph shown in Figure~\ref{figure:fig1}(a), the
outerplane depth of the only internal node is two, and the outerplane
depths of the other five nodes are all one. The outerplane radius of
the graph in Figure~\ref{figure:fig1}(a) is two and the outerplane
radius of the graph in Figure~\ref{figure:fig1}(c) is one.  All three
graphs in Figure~\ref{figure:fig1} are $2$-outerplane.  The graph in
Figure~\ref{figure:fig1}(c) is also $1$-outerplane.

\begin{figure}[t]
\centerline{\scalebox{0.85}{\input{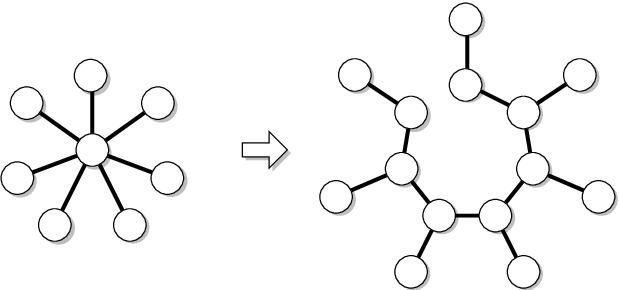}}}
\caption{The operation that turns a plane graph $G$ into
  $\reduce(G,v,u_1)$.}
\label{figure:fig2}
\end{figure}

Let $v$ be a node of plane graph $G$ with degree $d\geq 4$.  Let $u_1$
be a neighbor of $v$ in $G$.  For each $i=2,3,\ldots,d$, let $u_i$ be
the $i$-th neighbor of $v$ in $G$ starting from $u_1$ in clockwise
order around $v$.  Let $\reduce(G,v,u_1)$ be the plane graph obtained
from $G$ by the following steps, as illustrated by
Figure~\ref{figure:fig2}: (1) adding a zero-weight path
$(v_1,v_2,\ldots,v_d)$, (2) replacing each edge $(u_i,v)$ by edge
$(u_i,v_i)$ with $w(u_i,v_i)=w(u_i,v)$, and (3) deleting node $v$.

\begin{lemma}
\label{lemma:lemma2.5}
Let $v$ be a node of plane graph $G$ with degree four or more. If
$u_1$ is a neighbor of $v$ with the smallest outerplane depth in $G$,
then
\begin{enumerate}
\item
\label{item:lemma2.5:item1}
$\reduce(G,v,u_1)$ can be obtained from $G$ in time linear in the
degree of $v$ in $G$,
\item 
\label{item:lemma2.5:item2}
$\expand(\reduce(G,v,u_1))=\expand(G)$, and
\item 
\label{item:lemma2.5:item3}
$\radius(\reduce(G,v,u_1))=\radius(G)$.
\end{enumerate}
\end{lemma}
\begin{proof}
The first two statements are straightforward.  To prove the third
statement, let $j=\depth_G(v)$ and $G'=\reduce(G,v,u_1)$.  Let $G''$
be the plane graph obtained from $G'$ by peeling $j-1$ levels of
external nodes. By the choice of $u_1$, each $v_i$ with $1\leq i\leq
d$ is an external node in $G''$. Therefore, for each $i=1,2,\ldots,d$, we have
$\depth_{G'}(v_i)=j$.  Since the plane graphs obtained from $G$ and
$\reduce(G,v,u_1)$ by peeling $j$ levels of external nodes are
identical, the lemma is proved.
\end{proof}

\subsection{Proving the theorem by the main lemma}
\label{subsection:proving-main-theorem}
This subsection shows that, to prove Theorem~\ref{theorem:theorem1.1},
it suffices to ensure the following lemma.

\begin{lemma}
\label{lemma:lemma2.6}
If $G$ is an $O(1)$-degree plane graph satisfying the following equation
\begin{equation}
\wmax(G)+\radius(G)=O(\density(G))=O(\log^2 |G|),
\label{equation:lemma2.6}
\end{equation}
then $\girth(G)$ can be computed from $G$ in
$O(|G|+|\expand(G)|)$ time.
\end{lemma}

Now we prove Theorem~\ref{theorem:theorem1.1}.

\begin{proof}[Proof of Theorem~\ref{theorem:theorem1.1}]
Assume without loss of generality that the input $n$-node graph $G_0$
is biconnected.  Let $G$ be an $m$-node biconnected contracted planar
graph with $\expand(G)=G_0$ and $m\leq n$ that can be computed from
$G_0$ in $O(n)$ time, as ensured by
Lemma~\ref{lemma:lemma2.4}(\ref{item:lemma2.4:item1}).  By
Lemma~\ref{lemma:lemma2.2}(\ref{item:lemma2.2:item1}),
$\girth(G)=\girth(G_0)$.  If $n > m \log^2 m$, by
Theorem~\ref{lemma:lemma2.1}, it takes $O(m\log^2 m)=O(n)$ time to
compute $\girth(G)$.  The theorem is proved.  The rest of the proof
assumes $m\leq n \leq m\log^2 m$.

We first equip the $m$-node graph $G$ with a planar embedding, which
is obtainable in $O(m)$ time (see, e.g.,~\cite{BoyerM04}).
Initially, we have $|V(G)|=m$, $|V(\expand(G))|=n$, and
$\density(G)=\frac{n}{m}=O(\log^2 m)$.  We update $G$ in three
$O(m+n)$-time stages which maintain $|V(G)|=\Theta(m)$,
$|V(\expand(G))|=\Theta(n)$, $\girth(G)=\girth(G_0)$, and the
planarity of $G$.  At the end of the third stage, $G$ may contain
zero-weight edges and may no longer be biconnected and contracted.
However, the resulting $G$ is of degree at most three, has nonnegative
weights, and satisfies Equation~(\ref{equation:lemma2.6}).  The
theorem then follows from Lemma~\ref{lemma:lemma2.6}.

{\em Stage~1: Bounding the maximum weight of $G$.}  We repeatedly
replace $G$ by $\round(G, \lceil 36\cdot \density(G)\rceil)$ until
$\wmax(G)\leq \lceil 36\cdot \density(G)\rceil$ holds.  Although
$\density(G)$ may change in each iteration of the weight decreasing,
by Lemmas~\ref{lemma:lemma2.3}
and~\ref{lemma:lemma2.4}(\ref{item:lemma2.4:item2}) we know that
$\girth(G)$ remains the same and $\density(G)$ does not increase.
Since $G$ remains biconnected and contracted and has positive weights,
Lemma~\ref{lemma:lemma2.4}(\ref{item:lemma2.4:item2}) ensures
$\girth(G)\leq 36\cdot \density(G)$ throughout the stage.  After the
first iteration, $\wmax(G)\leq \lceil 36\cdot\frac{n}{m}\rceil$.  Each
of the following iterations decreases $\wmax(G)$ by at least one.
Therefore, this stage has $O(\frac{n}{m})$ iterations, each of which
takes $O(m)$ time, by
Lemma~\ref{lemma:lemma2.2}(\ref{item:lemma2.2:item2}).  The overall
running time is $O(n)$. The resulting $m$-node graph $G$ satisfies
\begin{math}
\wmax(G)=O(\density(G))=O(\log^2 |G|).
\end{math}

{\em Stage~2: Bounding the outerplane radius of $G$.}  For each
positive integer $j$, let $V_j$ consist of the nodes with outerplane
depths $j$ in $G$.  For each integer $i\geq 0$, let $G_i$ be the plane
subgraph of $G$ induced by the union of $V_j$ with $36\cdot i\cdot
\density(G)<j\leq 36\cdot (i+2)\cdot \density(G)$.  Let $G'$ be the
plane graph formed by the disjoint union of all the plane subgraphs
$G_i$ such that the external nodes of each $G_i$ remain external in
$G'$.  We have $\radius(G')=O(\density(G))$.  Each cycle of $G'$ is a
cycle of $G$, so $\girth(G)\leq \girth(G')$.  By
Lemma~\ref{lemma:lemma2.4}(\ref{item:lemma2.4:item2}), we have
$\girth(G)\leq 36\cdot\density(G)$.  Since the weight of each edge of
$G$ is at least one, the overlapping of the subgraphs $G_i$ in $G$
ensures that any cycle $C$ of $G$ with $w(C)=\girth(G)$ lies in some
subgraph $G_i$ of $G$, implying $\girth(G)\geq \girth(G')$.
Therefore, $\girth(G')=\girth(G)$.  By $|V(G')|=\Theta(|V(G)|)$ and
$|V(\expand(G'))|=\Theta(|V(\expand(G))|)$, we have
$\density(G')=\Theta(\density(G))$.  We replace $G$ by $G'$. The
resulting $G$ satisfies $\girth(G)=\girth(G_0)$ and
Equation~(\ref{equation:lemma2.6}).

{\em Stage~3: Bounding the degree of $G$.}  For each node $v$ of $G$
with degree four or more, we find a neighbor $u$ of $v$ in $G$ whose
outerplane depth in $G$ is minimized, and then replace $G$ by
$\reduce(G,v,u)$.  By
Lemma~\ref{lemma:lemma2.5}(\ref{item:lemma2.5:item1}), this stage
takes $O(m)$ time.  At the end, the degree of $G$ is at most three.
By Lemma~\ref{lemma:lemma2.5}(\ref{item:lemma2.5:item2}), the expanded
version of the resulting $G$ is identical to that of the $G$ at the
beginning of this stage.  By
Lemma~\ref{lemma:lemma2.5}(\ref{item:lemma2.5:item3}), the outerplane
radius remains the same.  The number of nodes in $G$ increases by at
most a constant factor.  The maximum weight remains the same.
Therefore, the resulting $G$ satisfies
Equation~(\ref{equation:lemma2.6}).  By
Lemma~\ref{lemma:lemma2.2}(\ref{item:lemma2.2:item1}), we have
$\girth(G)=\girth(G_0)$.
\end{proof}
The rest of the paper proves Lemma~\ref{lemma:lemma2.6}.

\section{Framework: dissection tree, nonleaf problem, and leaf problem}
\label{section:framework}

\begin{figure}[t]
\centerline{\scalebox{0.85}{\input{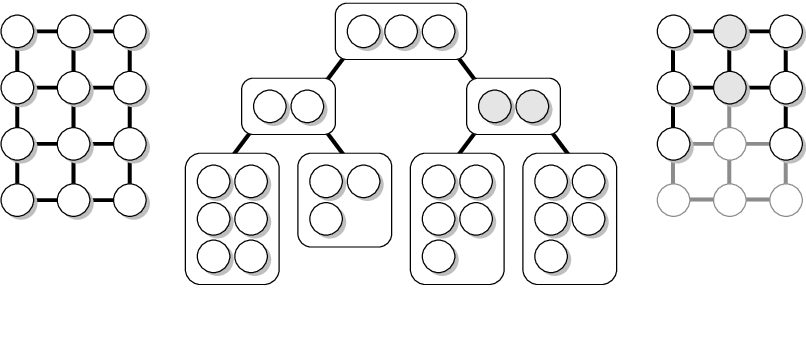}}}
\caption{(a) A weighted plane graph $G$. (b) A dissection tree $T$ 
  of $G$ with $S=\{7,8\}$ and $\border(S)=\{2,7,8,10\}$. 
  (c) Graph $G[\setbelow(S)]$. } 
\label{figure:fig3}
\end{figure}

This section shows the framework of our proof for
Lemma~\ref{lemma:lemma2.6}.  Let $G[S]$ denote the subgraph of $G$
induced by node set $S$.  Let $T$ be a rooted binary tree such that
each member of~$V(T)$ is a subset of $V(G)$.  To avoid confusion, we
use ``nodes'' to specify the members of $V(G)$ and ``vertices'' to
specify the members of $V(T)$.  Let $\treeroot(T)$ denote the root
vertex of $T$. Let $\leaf(T)$ consist of the leaf vertices
of~$T$. Let $\nonleaf(T)$ consist of the nonleaf vertices of~$T$.
For each vertex $S$ of~$T$, let $\setbelow(S)$ denote the union of
the vertices in the subtree of~$T$ rooted at $S$.  Therefore, if~$S$
is a leaf vertex of~$T$, then $\setbelow(S)=S$.  Also,
$\setbelow(\treeroot(T))$ consists of the nodes of $G$ that belong to
some vertex of~$T$.  For each nonleaf vertex $S$ of~$T$, let
$\lchild(S)$ and $\rchild(S)$ denote the two children of $S$ in
$T$. Therefore, if~$S$ is a nonleaf vertex of~$T$, then
\begin{math}
\setbelow(S) = S \cup \setbelow(\lchild(S)) \cup
\setbelow(\rchild(S)).
\end{math}
For instance, let $T$ be the tree in Figure~\ref{figure:fig3}(b).  We
have $\treeroot(T) = \{2,7,10\}$. Let $S=\rchild(\treeroot(T))$.  We
have $S=\{7,8\}$ and $\setbelow(S) = \{2,3,4,7,8,10,11,12\}$.  Let
$L=\lchild(S)$. We have $L=\setbelow(L)=\{2,3,4,7,8\}$.

Node sets $V_1$ and $V_2$ are {\em dissected} by node set $S$ in $G$
if any node in $V_1\setminus S$ and any node in $V_2\setminus S$ are
not adjacent in $G$.  We say that $T$ is a {\em dissection tree} of
$G$ if the following properties hold.
\begin{enumerate}[\quad\ \ $\bullet$\ \ \em Property 1:]
\item 
\label{property:prop1}
$\setbelow(\treeroot(T))=V(G)$.
\item 
\label{property:prop2}
The following statements hold for each nonleaf vertex $S$ of $T$.
\begin{enumerate}[(a)]
\item 
\label{property:prop2a}
$S\subseteq \setbelow(\lchild(S)) \cap \setbelow(\rchild(S))$.
\item 
\label{property:prop2b}
$\setbelow(\lchild(S))$ and $\setbelow(\rchild(S))$ are dissected by
$S$ in~$G$.
\end{enumerate}
\end{enumerate}
For instance, Figure~\ref{figure:fig3}(b) is a dissection tree of the
graph in Figure~\ref{figure:fig3}(a).

For any subset $S$ of $V(G)$, any two distinct nodes $u$ and $v$ of
$S$, and any edge $e$ of $G$, let $d_S(u,v;e)$ denote the distance of
$u$ and $v$ in $G[\setbelow(S)]\setminus \{e\}$ and let $d_S(u,v)$
denote the distance of $u$ and $v$ in $G[\setbelow(S)]$.  Observe that
if $e_S(u,v)$ is an edge in some min-weight path between $u$ and $v$
in $G[\setbelow(S)]$, then $d_S(u,v)+d_S(u,v;e_S(u,v))$ is the minimum
weight of any non-degenerate cycle in $G[\setbelow(S)]$ containing $u$
and $v$ that traverses $e_S(u,v)$ exactly once.  For instance, let $G$
and $T$ be shown in Figure~\ref{figure:fig3}(a) and
\ref{figure:fig3}(b).  If $S = \{7,8\}$, then $G[\setbelow(S)]$ is as
shown in Figure~\ref{figure:fig3}(c).  We have $d_S(7,10) = 7$ (e.g.,
path $(7,8,12,11,10)$ has weight~$7$) and $d_S(7,10;(7,8)) = 10$
(e.g., path $(7,3,4,8,12,11,10)$ has weight $10$).  Since $(7,8)$ is
an edge in a min-weight path $(7,8,12,11,10)$ between nodes $7$ and
$10$, the minimum weight of any non-degenerate cycle in
$G[\setbelow(S)]$ containing nodes $7$ and $10$ that traverses $(7,8)$
exactly once is~$17$ (e.g., non-degenerate
cycle~$(7,8,12,11,10,11,12,8,4,3,7)$ has weight~$17$ and traverses
$(7,8)$ exactly once).

\begin{definition}
{\em For any dissection tree $T$ of graph $G$, the {\em nonleaf
    problem} of $(G,T)$ is to compute the following information for
  each nonleaf vertex $S$ of $T$ and each pair of distinct nodes $u$
  and $v$ of $S$: (1) an edge $e_S(u,v)$ in a min-weight path between
  $u$ and $v$ in $G[\setbelow(S)]$ and (2) distances $d_S(u,v)$ and
  $d_S(u,v;e_S(u,v))$.  }
\end{definition}

\begin{definition}
{\em For any dissection tree $T$ of graph $G$, the {\em leaf problem}
  of $(G,T)$ is to compute the minimum $\girth(G[L])$ over all leaf
  vertices $L$ of $T$.}
\end{definition}

Define the {\em sum of squares} of a dissection tree $T$ as
$$\squares(T)=\sum_{S\in \nonleaf(T)}|S|^2.$$ Our proof for
Lemma~\ref{lemma:lemma2.6} consists of the following three tasks.
\begin{itemize}
\item {\em Task 1}. Computing a dissection tree $T$ of $G$ with
  $\squares(T)=O(|G|)$.

\item {\em Task 2}. Solving the nonleaf problem of $(G,T)$.

\item {\em Task 3}. Solving the leaf problem of $(G,T)$.
\end{itemize}
The following lemma ensures that, to prove Lemma~\ref{lemma:lemma2.6},
it suffices to complete all three tasks in $O(|G|+|\expand(G)|)$ time
for any $O(1)$-degree plane graph $G$ satisfying
Equation~(\ref{equation:lemma2.6}).

\begin{lemma}
\label{lemma:lemma3.1}
Given a dissection tree $T$ of graph $G$ and solutions to the leaf
and nonleaf problems of $(G,T)$, it takes $O(\squares(T))$ time to
compute $\girth(G)$.
\end{lemma}

\begin{proof}
Let $g_{\textit{leaf}}$ be the given solution to the leaf problem of
$(G,T)$.  It takes $O(\squares(T))$ time to compute the minimum
value $g_{\textit{nonleaf}}$ of $d_S(u,v)+d_S(u,v;e_S(u,v))$ over all
pairs of distinct nodes $u$ and $v$ of $S$, where $e_S(u,v)$ is the
edge in the given solution to the nonleaf problem of $(G,T)$.  Let
$C$ be a simple cycle of $G$ with $w(C)=\girth(G)$.  It suffices to
show $w(C) = \min\{g_\textit{leaf}, g_\textit{nonleaf}\}$.  By
Property~\ref{property:prop1} of $T$, there is a lowest vertex $S$ of
$T$ with $V(C)\subseteq \setbelow(S)$. If $S$ is a leaf vertex of
$T$, then $w(C)=g_{\textit{leaf}}$.  If $S$ is a nonleaf vertex of
$T$, then $w(C)=\girth(G[\setbelow(S)])$.  We know $|S\cap V(C)|\geq
2$: Assume $|S\cap V(C)|\leq 1$ for contradiction. By
Property~\ref{property:prop2b} and simplicity of $C$, we have
$V(C) \subseteq S \cup \lchild(S)$ or $V(C) \subseteq S \cup
\rchild(S)$.  By Property~\ref{property:prop2a}, either $V(C)
\subseteq \lchild(S)$ or $V(C) \subseteq \rchild(S)$ holds,
contradicting the choice of $S$.  Let $u$ and $v$ be two distinct
nodes in $S\cap V(C)$.  Since $C$ is a min-weight non-degenerate cycle
of $G[\setbelow(S)]$, we have $w(C)=d_S(u,v)+d_S(u,v;e_S(u,v))$.
Therefore, $w(C)=g_{\textit{nonleaf}}$.  The lemma is proved.
\end{proof}

\section{Task 1: computing a dissection tree}
\label{section:task1}
Let $T$ be a dissection tree of graph $G$.  For each vertex $S$ of
$T$, let $\setabove(S)$ be the union of the ancestors of $S$ in $T$
and let $\inherit(S) = \setabove(S) \cap \setbelow(S)$.  If $S$ is a
leaf vertex of $T$, then let $\border(S)=\inherit(S)$.  If $S$ is a
nonleaf vertex of $T$, then let $\border(S)=S\cup \inherit(S)$.  For
instance, let $T$ be as shown in Figure~\ref{figure:fig3}(b).  Let
$S=\rchild(\treeroot(T))$. We have
$\setabove(S)=\inherit(S)=\{2,7,10\}$ and $\border(S)=\{2,7,8,10\}$.
Let $L=\lchild(S)$. We have $\setabove(L)=\{2,7,8,10\}$ and
$\inherit(L)=\border(L)=\{2,7,8\}$.
Define
\begin{displaymath}
\ell(m)=\lceil\log^{30} m\rceil.
\end{displaymath}
For any positive integer $r$, a dissection tree $T$ of an $m$-node
graph $G$ is an {\em $r$-dissection tree} of $G$ if the following 
conditions hold.
\begin{enumerate}[\quad\ \ $\bullet$\ \ \em {Condition} 1:]
\item 
\label{condition:cond1}
$|V(T)|=O(m/\ell(m))$ and $\sum_{L \in \leaf(T)}|\border(L)|=O(mr/\ell(m))$.
\item 
\label{condition:cond2}
$|L| = \Theta(\ell(m))$ and $|\border(L)| = O(r\log m)$ hold for each
leaf vertex $L$ of~$T$.

\item 
\label{condition:cond3}
$|S|+|\border(S)| = O(r\log m)$ hold for each nonleaf vertex $S$ of
$T$.
\end{enumerate}
For any $r$-outerplane $G$, it takes $O(m)$ time to compute an
$O(r)$-node set $S$ of $G$ such that the node subsets $V_1$ and $V_2$
of $G$ dissected by $S$ satisfy
$|V_1|/|V_2|=\Theta(1)$ (see, e.g., \cite{RobertsonS84,Bodlaender98}).  
By recursively applying this linear-time procedure, an $r$-dissection
tree can be obtained in $O(m\log m)$ time, which is too expensive for
our algorithm.  Instead, based upon Goodrich's $O(m)$-time separator
decomposition~\cite{Goodrich95}, we prove the following lemma.

\begin{lemma}
\label{lemma:lemma4.1}
Let $G$ be an $m$-node $r$-outerplane $O(1)$-degree graph with
$r=O(\log^2 m)$.  It takes $O(m)$ time to compute an $r$-dissection
tree of $G$.
\end{lemma}

\begin{figure}[t]
\centerline{\scalebox{0.85}{\input{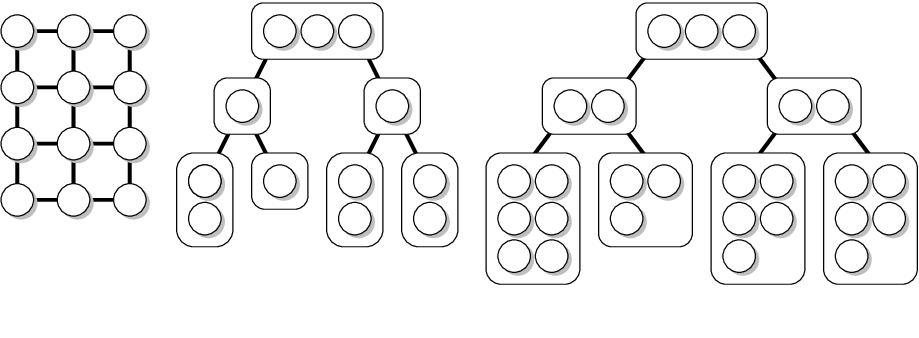}}}
\caption{(a) A plane graph $G$.  (b) A decomposition tree $T'$ of $G$.
(c) A dissection tree $T$ of $G$.}
\label{figure:fig4}
\end{figure}

Let $T'$ be a rooted binary tree such that each
vertex of $T'$ is a subset of $V(G)$.  We say that $T'$ is a {\em
  decomposition tree} of $G$ if Properties~\ref{property:prop1}
and~\ref{property:prop2b} hold for $T'$.  For instance, 
Figure~\ref{figure:fig4}(b) shows a decomposition tree of the graph in
Figure~\ref{figure:fig4}(a).  For any $m$-node triangulated plane
graph $\Delta$ and for any positive integer $\ell\leq m$,
Goodrich~\cite{Goodrich95} showed that it takes $O(m)$ time to compute
an $O(m/\ell)$-vertex $O(\log m)$-height decomposition tree $T'$ of
$\Delta$ such that $|L| = \Theta(\ell)$ holds for each leaf vertex $L$
of $T'$ and $|S|=O(|\setbelow(S)|^{0.5})$ holds for nonleaf vertex
$S$ of $T'$.  As a matter of fact, Goodrich's techniques directly
imply that if an $O(r)$-diameter spanning tree of $\Delta$ is given,
then a decomposition tree $T'$ of $\Delta$ satisfying the following four
conditions can also be obtained efficiently.
\begin{enumerate}[\quad\ \ $\bullet$\ \ \em {Condition} 1':]
\item 
\label{condition:cond1x}
$|V(T')|=O(m/\ell(m))$.
\item 
\label{condition:cond2x}
$|L| = \Theta(\ell(m))$ and $|\border(L)|=0$ hold for each leaf vertex
$L$ of $T'$.
\item 
\label{condition:cond3x}
$|S|=|\border(S)|=O(r)$ holds for each nonleaf vertex $S$ of $T'$.
\item 
\label{condition:cond4x}
The height of $T'$ is $O(\log m)$.
\end{enumerate}
\begin{lemma}
\label{lemma:lemma4.2}
Given an $O(r)$-diameter spanning tree of an $m$-node simple
triangulated plane graph $\Delta$ with $r=O(\log^2 m)$, it takes
$O(m)$ time to compute a decomposition tree $T'$ of $\Delta$ that
satisfies Properties~\ref{property:prop1} and~\ref{property:prop2b}
and
Conditions~\ref{condition:cond1x}',~\ref{condition:cond2x}',~\ref{condition:cond3x}',
and~\ref{condition:cond4x}'.
\end{lemma}
\begin{proof}
The lemma can be proved by following what Goodrich did
in~\cite{Goodrich95}, so we give only a proof sketch here.
Goodrich~\cite[\S2.4]{Goodrich95} showed that, with some $O(m)$-time
pre-computable dynamic data structures for the given $O(r)$-diameter
spanning tree and $\Delta$, it takes $O(r\log^{O(1)} m)$ time to find
a fundamental cycle $C$ of $\Delta$ with respect to the given spanning
tree such that the maximum number of nodes either inside or outside
$C$ is minimized.  Since the diameter of the given spanning tree is
$O(r)$, we have $|C|=O(r)$.  Let $V_1$ (respectively, $V_2$) consist
of the nodes of $\Delta$ inside (respectively, outside)~$C$.  We have
$|V_1|/|V_2|=\Theta(1)$, as shown by Lipton and
Tarjan~\cite{LiptonT79}.  With the pre-computed data structures, it
also takes $O(r\log^{O(1)} m)$ time to (1) split $\Delta$ into
$\Delta[V_1]$ and $\Delta[V_2]$ and (2) split the given
$O(r)$-diameter spanning tree of $\Delta$ into an $O(r)$-diameter
spanning tree of $\Delta[V_1]$ and an $O(r)$-diameter spanning tree of
$\Delta[V_2]$.  Let $T'$ be obtained by recursively computing
$O(r)$-node sets $\lchild(S)$ and $\rchild(S)$ of $\Delta[V_1]$ and
$\Delta[V_2]$ until $|S|\leq \ell(m)$.  As long as
$r=O(m^{1-\epsilon})$ holds for some constant $\epsilon>0$, the
overall running time is $O(m)$.  One can verify that the resulting
tree $T'$ indeed satisfies Properties~\ref{property:prop1}
and~\ref{property:prop2b} and
Conditions~\ref{condition:cond1x}',~\ref{condition:cond2x}',~\ref{condition:cond3x}',
and~\ref{condition:cond4x}'.
\end{proof}

We prove Lemma~\ref{lemma:lemma4.1} using Lemma~\ref{lemma:lemma4.2}.

\begin{proof}[Proof of Lemma~\ref{lemma:lemma4.1}]
It takes $O(m)$ time to triangulate the $m$-node $r$-outerplane graph
$G$ into an $m$-node simple triangulated plane graph $\Delta$ that
admits a spanning tree with diameter $O(r)$.  Specifically, we first
triangulate each connected component of $G$ into a simple biconnected
internally triangulated plane graph $G'$ such that the outerplane depth of
each node remains the same after the triangulation.  
Let $u_0$ be an arbitrary external node of $G'$. We
then add an edge $(u_0,u)$ for each external node $u$ of $G'$ that is
not adjacent to $u_0$. The resulting graph $\Delta$ is an $m$-node
$O(r)$-outerplane simple triangulated plane graph. 
An $O(r)$-diameter spanning tree of
$\Delta$ can be obtained in $O(m)$ time as follows. Let $u_0$ be the
parent of all of its neighbors in $\Delta$.  For each node $u$ other
than $u_0$ and the neighbors of $u_0$, we arbitrary choose a neighbor
$v$ of $u$ in $\Delta$ with $\depth_\Delta(v)=\depth_\Delta(u)-1$ and
let $v$ be the parent of $u$ in the spanning tree. The diameter of the
resulting spanning tree of $\Delta$ is $O(r)$.  For instance, let $G$
be as shown in Figure~\ref{figure:fig5}(a).  An example of $G'$ is
shown in Figure~\ref{figure:fig5}(b).  An example of $\Delta$ together
with its spanning tree rooted at $u_0$ is shown in
Figure~\ref{figure:fig5}(c).

\begin{figure}[t]
\centerline{\scalebox{0.85}{\input{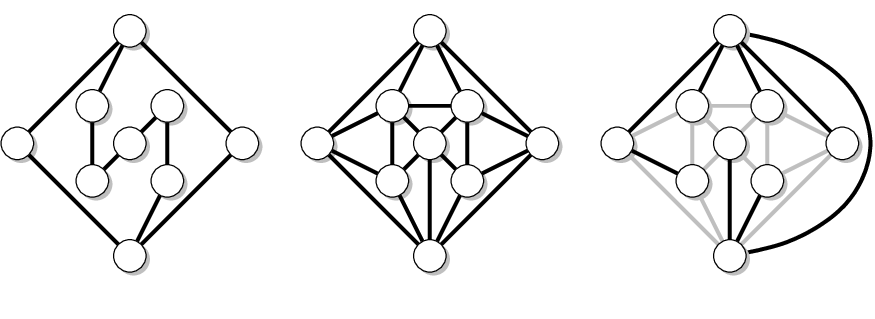}}}
\caption{(a) A plane graph $G$. Each node is labeled by its outerplane
  depth.  (b) A biconnected internally triangulated plane graph $G'$
  obtained from $G$.  (c) A triangulated plane graph $\Delta$ obtained
  from $G'$ with a spanning tree of $\Delta$ rooted at $u_0$.}
\label{figure:fig5}
\end{figure}

Let $T'$ be a decomposition tree of $\Delta$ as ensured by
Lemma~\ref{lemma:lemma4.2}. Since $\Delta$ is obtained from $G$ by
adding edges, $T'$ is also a decomposition tree of $G$ that satisfies
Properties~\ref{property:prop1} and~\ref{property:prop2b} and
Conditions~\ref{condition:cond1x}',~\ref{condition:cond2x}',~\ref{condition:cond3x}',
and~\ref{condition:cond4x}'.  We prove the lemma by showing that $T'$
can be modified in $O(m)$ time into an $r$-dissection tree $T$ of $G$
by calling $\descend(\treeroot(T'))$, where the recursive procedure
$\descend(S)$ is defined as follows.  If $S$ is a leaf vertex of
$T'$, then we return. If $S$ is a nonleaf vertex of $T'$, we first
(1) run the following steps for each node $u$ of the current $S$, and
then (2) recursively call $\descend(\lchild(S))$ and $\descend(\rchild(S))$.
\begin{enumerate}[\em Step 1.]
\item If $u$ is not adjacent to any node in the current $\setbelow(\lchild(S))$ in $G$,
  then we delete $u$ from $S$ and insert $u$ into the current $\rchild(S)$.

\item If $u$ is adjacent to some node in the current $\setbelow(\lchild(S))$ in $G$
  and is not adjacent to any node in the current $\setbelow(\rchild(S))$ in $G$,
  then we delete $u$ from $S$ and insert $u$ into the current
  $\lchild(S)$.

\item If $u$ is adjacent to some node in the current $\setbelow(\lchild(S))$ and
  some node in the current $\setbelow(\rchild(S))$ in $G$, then we leave $u$ in $S$
  and insert $u$ into the current $\lchild(S)$ and $\rchild(S)$.
\end{enumerate}
For instance, if the decomposition tree $T'$ is as shown in
Figure~\ref{figure:fig4}(b), then the resulting tree $T$ of running
$\descend(\treeroot(T'))$ is as shown in Figure~\ref{figure:fig4}(c).

We show that $T$ is indeed an $r$-dissection tree of $G$.  By definition 
of $\descend$, one can verify that a node $u$ belongs to a nonleaf vertex 
$S$ of $T$ if and only if $u$ belongs to both $\setbelow(\lchild(S))$ and
$\setbelow(\rchild(S))$ in $T$.  Property~\ref{property:prop2a} holds
for $T$ and, thereby, Properties~\ref{property:prop1}
and~\ref{property:prop2} of $T$ follow from
Properties~\ref{property:prop1} and~\ref{property:prop2b} of $T'$.
Moreover, if $u$ belongs to a nonleaf vertex $S$ of $T$, then 
the degrees of $u$ in $G[\setbelow(\lchild(S))]$ and
$G[\setbelow(\rchild(S))]$ are strictly less than the degree of $u$ in
$G[\setbelow(S)]$.  Since the degree of $G$ is $O(1)$, each node $u$
of $G$ belongs to $O(1)$ vertices of $T$. By
Conditions~\ref{condition:cond1x}' and~\ref{condition:cond3x}' of
$T'$, we have $\sum_{L\in\leaf(T)}|\border(L)|=
\sum_{S\in\nonleaf(T')}O(|S|)=O(mr/\ell(m))$ and
$|V(T)|=|V(T')|=O(m/\ell(m))$.  Condition~\ref{condition:cond1} of
$T$ holds.  By Conditions~\ref{condition:cond3x}'
and~\ref{condition:cond4x}' of $T'$, the procedure increases $|S|$
and $|\border(S)|$ for each vertex $S$ of $T'$ by $O(r\log m)$.
Therefore, Conditions~\ref{condition:cond2} and~\ref{condition:cond3}
of $T$ follow from Conditions~\ref{condition:cond2x}'
and~\ref{condition:cond3x}' of $T'$.

We show that $T$ can be obtained from $T'$ in $O(m)$ time.  We first
spend $O(m)$ time to compute for each node $v$ of $G$ a list of $O(1)$
vertices of the original $T'$ that contain $v$.  Consider the case
that $S$ is a nonleaf vertex of the current $T'$. Let $S'$ be a child
vertex of $S$ in the current $T'$.  To determine whether a node $u$ of
$S$ is adjacent to some node in the current $\setbelow(S')$, for all
$O(1)$ neighbors $v$ of $u$ in $G$, we traverse upward in $T'$ from
the $O(1)$ vertices of $T'$ that currently contain $v$. The traversal
passes $S'$ if and only if $u$ is adjacent to some node in the current
$\setbelow(S')$.  By Condition~\ref{condition:cond4x}' of $T'$, it
takes $O(\log m)$ time to determine whether $u$ is adjacent to the
current $\setbelow(S')$.  Each update to the list of vertices of $T'$
that contains $u$ takes $O(1)$ time.  By
Conditions~\ref{condition:cond1x}',~\ref{condition:cond3x}',
and~\ref{condition:cond4x}' of $T'$, the overall running time of
$\descend(\treeroot(T'))$ is $O(mr\log^2 m/\ell(m))=O(m)$. The lemma
is proved.
\end{proof}

\section{Task 2: solving the nonleaf problems}
\label{section:task2}
This section proves the following lemma.
\begin{lemma}
\label{lemma:lemma5.1}
Let $G$ be an $m$-node $O(1)$-degree $r$-outerplane graph with
$\wmax(G)+r=O(\log^2 m)$.  Given an $r$-dissection tree $T$ of $G$,
the nonleaf problem of $(G,T)$ can be solved in $O(mr)$ time.
\end{lemma}

\begin{definition}
{\em Let $T$ be a dissection tree of $G$.  Let $S$ be a vertex of
  $T$.  The {\em border problem} of $(G,T)$ for $S$ is to compute
  the following information for any two distinct nodes $u$ and $v$ of
  $\border(S)$: (1) $d_S(u,v)$, (2) an edge $e_S(u,v)$ on some
  min-weight path between $u$ and $v$ in $G[\setbelow(S)]$ that is
  incident to $u$, and (3) $d_S(u,v;e)$ for each edge $e$ of $G$
  incident to $u$.  }
\end{definition}
Since $S\subseteq \border(S)$ holds for each nonleaf vertex $S$ of
$T$, any collection of solutions to the border problems of $(G,T)$
for all nonleaf vertices of $T$ yields a solution to the nonleaf
problem of $(G,T)$.  We prove Lemma~\ref{lemma:lemma5.1} by solving
the border problems of $(G,T)$ for all vertices of $T$ in $O(mr)$
time.
A leaf vertex $L$ in an $r$-dissection tree $T$ of an $m$-node graph
$G$ is {\em special} if
\begin{displaymath}
|\border(L)|+r\leq \lceil \log^2\ell(m)
\rceil.
\end{displaymath}
Section~\ref{subsection:nonleaf} shows that the border problems of
$(G,T)$ for all vertices of $T$ can be reduced in $O(mr)$ time to the
border problems of $(G,T)$ for all special leaf vertices of $T$, as
summarized by Lemma~\ref{lemma:lemma5.4}.
Section~\ref{subsection:special} shows that the border problems of
$(G,T)$ for all special leaf vertices of $T$ can be solved in $O(mr)$
time, as summarized by Lemma~\ref{lemma:lemma5.6}.
Lemma~\ref{lemma:lemma5.1} follows immediately from
Lemmas~\ref{lemma:lemma5.4} and~\ref{lemma:lemma5.6}.

\subsection{A reduction to the border problems for the special leaf vertices}
\label{subsection:nonleaf}

Our reduction uses the following dynamic data structure that supports
distance queries.

\begin{lemma}[Klein~\cite{Klein05}]
\label{lemma:lemma5.2}
Let $G$ be an $\ell$-node planar graph.  It takes $O(\ell\log^2 \ell)$
time to compute a data structure $\oracle(G)$ such that each update to the
weight of an edge and each query to the distance between any two nodes
in $G$ can be supported by $\oracle(G)$ in time $O(\ell^{2/3}\log^{5/3}
\ell)=O(\ell^{7/10})$.
\end{lemma}

The following lemma is needed to ensure the correctness of our reduction 
via dynamic programming.

\begin{lemma}
\label{lemma:lemma5.3}
For each nonleaf vertex $S$ of $T$, we have $S\subseteq
\border(\lchild(S))\cap \border(\rchild(S))$ and $\border(S)\subseteq
\border(\lchild(S))\cup \border(\rchild(S))$.
\end{lemma}
\begin{proof}
Let $S'=\lchild(S)$ and $S''=\rchild(S)$.  By
Property~\ref{property:prop2a} of $T$, $S\subseteq \setbelow(S')\cap
\setbelow(S'')$.  By $S\subseteq \setabove(S')\cap\setabove(S'')$, we
have $S\subseteq \inherit(S')\cap\inherit(S'')$.  By
$\inherit(S')\subseteq\border(S')$ and
$\inherit(S'')\subseteq\border(S'')$, we have
$S\subseteq\border(S')\cap\border(S'')$.  We also have
\begin{eqnarray*}
\inherit(S)\setminus S
&=&((\setbelow(S')\cup\setbelow(S'')\cup S)\cap \setabove(S))\setminus S\\
&\subseteq&(\setbelow(S')\cup\setbelow(S''))\cap \setabove(S)\\
&=&(\setbelow(S') \cap \setabove(S))
\cup (\setbelow(S'') \cap \setabove(S))\\
&\subseteq&(\setbelow(S') \cap \setabove(S')) 
\cup (\setbelow(S'') \cap \setabove(S''))\\
&=&\inherit(S')\cup\inherit(S'')\\
&\subseteq&\border(S')\cup\border(S'').
\end{eqnarray*}
Thus, $\border(S)=S\cup (\inherit(S)\setminus
S)\subseteq \border(S')\cup\border(S'')$.  The lemma
is proved.
\end{proof}

The following lemma shows the reduction.

\begin{lemma}
\label{lemma:lemma5.4}
Let $G$ be an $m$-node $O(1)$-degree graph.  Given (1) an
$r$-dissection tree $T$ of $G$ with $r=O(\log^2 m)$ and (2) solutions
to the border problems of $(G,T)$ for all special leaf vertices of
$T$, it takes $O(mr)$ time to solve the border problems of $(G,T)$
for all vertices of $T$.
\end{lemma}

\begin{proof}
Solutions for special leaf vertices are given.  We first show that it
takes $O(mr)$ time to compute solutions for all non-special leaf
vertices $L$ of $T$.  Let $\ell=\ell(m)$.  By
Condition~\ref{condition:cond1} of $T$, we have $\sum_{L\in
  \leaf(T)} (|\border(L)|+r)=O(mr/\ell)$, implying that $T$ has
$O(\frac{mr}{\ell\log^2 \ell})$ non-special leaf vertices.  For each
non-special leaf vertex $L$ of $T$, we run the following
$O(\ell\log^2 \ell)$-time steps.
\begin{enumerate}[\em Step 1.]
\item 
By Condition~\ref{condition:cond2} of $T$, we have
$|L|=\Theta(\ell)$.  We compute a data structure $\oracle(G[L])$ in
$O(\ell\log^2\ell)$ time as ensured by Lemma~\ref{lemma:lemma5.2}.

\item
For any two nodes $u$ and $v$ in $\border(L)$, we first obtain
$d_L(u,v)$ from $\oracle$ in $O(\ell^{7/10})$ time.  We then find a neighbor
$x$ of $u$ in $G[L]$ with $d_L(u,v)=w(u,x)+d_L(x,v)$ and let
$e_L(u,v)=(u,x)$, which can be obtained from $\oracle$ in $O(\ell^{7/10})$
time, since the degree of $G$ is $O(1)$.  By
Lemma~\ref{lemma:lemma5.2} and Condition~\ref{condition:cond2} of
$T$, the overall time complexity for this step is $O(\ell^{7/10}\cdot
|\border(L)|^2)=O(\ell^{7/10}\cdot r^2\log^2 m)=O(\ell^{9/10})$.

\item 
For each edge $e$ that is incident to $\border(L)$, we compute
$d_L(u,v;e)$ from $\oracle$ for all nodes $u$ and $v$ of $\border(L)$ as
follows: (1) Temporarily setting $w(e)=\infty$; (2) for each pair of
distinct nodes $u$ and $v$ in $\border(L)$, obtaining $d_L(u,v;e)$
from the distance of $u$ and $v$ in the current $G[L]$; and (3)
restoring the original weight of $e$.  Since the degree of $G$ is
$O(1)$, there are $O(|\border(L)|)$ choices of $e$.  By
Lemma~\ref{lemma:lemma5.2} and Condition~\ref{condition:cond2} of
$T$, the running time of this step is $O(\ell^{7/10}\cdot
|\border(L)|^3) =O(\ell^{7/10}\cdot r^3 \log^3 m) = O(\ell)$.
\end{enumerate}

\begin{figure}[t]
\centerline{\scalebox{0.85}{\input{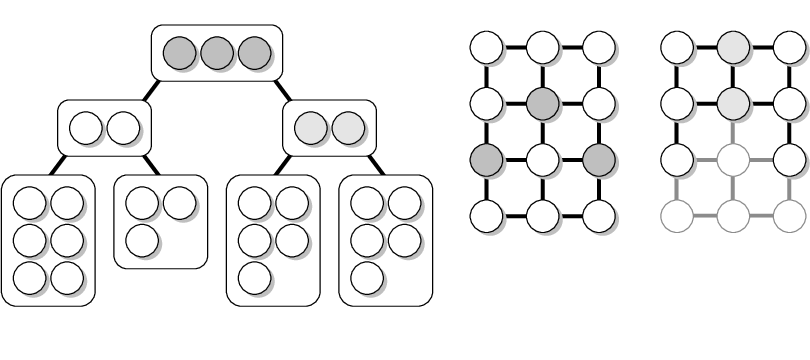}}}
\caption{(a) A dissection tree $T$ of the graph in~(b) with $R =
  \border(R)=\{2,7,10\}$, $S=\{7,8\}$, and $\border(S)=\{2,7,8,10\}$.
  (b) Graph $G = G[\setbelow(R)]$.  (c) Graph $G[\setbelow(S)]$.}
\label{figure:fig6}
\end{figure}

We now show that the solutions for all nonleaf vertices $S$ of $T$
can be computed in $O(m)$ time.  By definition of $\ell(m)$ and
Condition~\ref{condition:cond1} of $T$, we have
$|\nonleaf(T)|=O(m/\log^{30} m)$.  By $r=O(\log^2 m)$ and
Condition~\ref{condition:cond3} of $T$, we have
$|S|+|\border(S)|=O(\log^3 m)$.  It suffices to prove the following
{\em claim} for each nonleaf vertex $S$ of $T$: ``Given solutions for
$S'=\lchild(S)$ and $S''=\rchild(S)$, a solution for $S$ can be
computed in $O(|\border(S)|^3\cdot |S|^2)$ time.''  By
Property~\ref{property:prop2b} of $T$, $\setbelow(S')$ and
$\setbelow(S'')$ are dissected by $S$ in~$G$.  We use $(S,k)$-path to
denote a path of $G[\setbelow(S)]$ that switches to a different side
of $S$ at most $k$ times: Precisely, an {\em $(S,0)$-path} is a path
that completely lies in $G[\setbelow(S')]$ or completely lies in
$G[\setbelow(S'')]$.  For any positive integer $k$, we say that
$(u_1,u_2,\ldots,u_t)$ is an {\em $(S,k)$-path} if
$(u_1,u_2,\ldots,u_{t'})$ is an $(S,k-1)$-path, where $t'$ is the
smallest integer such that $(u_{t'},u_{t'+1},\ldots,u_t)$ is an
$(S,0)$-path.  For instance, let $T$ and $G$ be as shown in
Figures~\ref{figure:fig6}(a) and \ref{figure:fig6}(b). Let
$S=\{7,8\}$.  Note that $(8,7,11,10)$ is both an $(S,0)$-path and an
$(S,1)$-path with $u_{t'} = 8$.  However, $(2,3,7,11,10)$ is an
$(S,1)$-path with $u_{t'} = 7$ but not an $(S,0)$-path.  Based upon
the facts $\border(S)\subseteq\border(S')\cup\border(S'')$ and
$S\subseteq\border(S')\cap\border(S'')$ as ensured by
Lemma~\ref{lemma:lemma5.3}, we prove the above claim in the following
three stages, each of which is also illustrated by
Figure~\ref{figure:fig6}.
\begin{enumerate}[\em {Stage} 1.]
\item 
For any nodes $u$ and $v$ in $\border(S)$, let
$d_{S,i}(u,v)$ denote the minimum weight of any $(S,i)$-path of
$G[\setbelow(S)]$ between $u$ and $v$.  Any simple path of
$G[\setbelow(S)]$ is an $(S,|S|)$-path, so $d_S(u,v)=d_{S,|S|}(u,v)$.
As illustrated by~Figure~\ref{figure:fig6}(b), we have 
$d_{R,0}(10,2) = 7$ and $d_{R,1}(10,2) = 4$.
As illustrated by Figure~\ref{figure:fig6}(c), we have 
$d_{S,0}(10,2) = \infty$ and $d_{S,1}(10,2) = 9$.
One can verify the following recurrence relation.
\begin{displaymath}
d_{S,i}(u,v)=
\left\{
\begin{array}{ll}
0&\mbox{if $i=0$ and $u=v$};\\
\min\{d_{S'}(u,v), d_{S''}(u,v)\}&\mbox{if $i=0$ and $u\neq v$};\\
\min\{d_{S,i-1}(u,y)+d_{S,0}(y,v): y\in S\cup\{v\}\}&\mbox{if $i\geq 1$}.
\end{array}
\right.
\end{displaymath}
This stage takes $O(|\border(S)|^2 \cdot |S|^2)$ time via
dynamic programming.

\item 
For any distinct nodes $u$ and $v$ in $\border(S)$, let
$e_{S,i}(u,v)$ denote an incident edge of $u$ in a min-weight
$(S,i)$-path of $G[\setbelow(S)]$ between $u$ and $v$.  
If no $(S,i)$-path of $G[\setbelow(S)]$ between $u$ and $v$
exists, let $e_{S,i}(u,v) = \varnothing$.
As illustrated by~Figure~\ref{figure:fig6}(b), edge $(10,6)$ 
is the only choice for $e_{R,0}(10,2)$ and $e_{R,1}(10,2)$.
As illustrated by Figure~\ref{figure:fig6}(c), we have 
$e_{S,0}(10,2) = \varnothing$, and edge $(10,11)$ is the 
only choice for $e_{S,1}(10,2)$.
Let
\begin{displaymath}
e_{S,i}(u,v)=
\left\{
\begin{array}{ll}
e_{S'}(u,v)&\mbox{if $i=0$ and $d_{S'}(u,v)\leq d_{S''}(u,v)$};\\
e_{S''}(u,v)&\mbox{if $i=0$ and $d_{S'}(u,v)>d_{S''}(u,v)$};\\
e_{S,i-1}(u,y)&\mbox{if $i\geq 1$},
\end{array}
\right.
\end{displaymath}
where $y$ can be any node in $S\cup\{v\}\setminus \{u\}$ with
$d_{S,i}(u,v)=d_{S,i-1}(u,y)+d_{S,0}(y,v)$.  Since both $e_{S'}(u,v)$
and $e_{S''}(u,v)$ are incident to $u$ in $G[\setbelow(S)]$, each
$e_{S,i}(u,v)$ is incident to $u$ in $G[\setbelow(S)]$.  Therefore,
$e_{S,|S|}(u,v)$ is a valid choice of $e_S(u,v)$.  This stage takes
$O(|\border(S)|^2\cdot |S|^2)$ time via dynamic programming.

\item 
For any nodes $u$ and $v$ in $\border(S)$ and any edge
$e$ of $G[\setbelow(S)]$ that is incident to $\border(S)$, let
$d_{S,i}(u,v;e)$ be the minimum weight of any $(S,i)$-path in
$G[\setbelow(S)]\setminus \{e\}$ between $u$ and $v$.  We have
$d_S(u,v;e)=d_{S,|S|}(u,v;e)$.  
As illustrated by Figure~\ref{figure:fig6}(b), we have 
$d_{R,0}(10,2;(10,6)) = d_{R,1}(10,2;(10,6)) = 8$.
As illustrated by Figure~\ref{figure:fig6}(c), we have 
$d_{S,0}(10,2;(10,11)) = d_{S,1}(10,2;(10,11)) = \infty$.  
One can verify the following recurrence relation.
\begin{displaymath}
d_{S,i}(u,v;e)=\left\{
\begin{array}{ll}
0&\mbox{if $i=0$ and $u=v$};\\
\min\{d_{S'}(u,v;e), d_{S''}(u,v;e)\}&\mbox{if $i=0$ and $u \neq v$};\\
\min\{d_{S,i-1}(u,y;e)+d_{S,0}(y,v;e): y\in S\cup\{v\}\}&\mbox{if $i\geq 1$}.
\end{array}
\right.
\end{displaymath}
Since the degree of $G$ is $O(1)$, the number of choices of $e$ is
$O(|\border(S)|)$.  This stage takes $O(|\border(S)|^3\cdot |S|^2)$
time via dynamic programming.
\end{enumerate}
The lemma is proved.
\end{proof}

\subsection{Solving the border problems for the special leaf vertices}
\label{subsection:special}

We need the following linear-time pre-computable data structure in the
proof of Lemma~\ref{lemma:lemma5.6} to solve the border problems of
$(G,T)$ for all special leaf vertices of $T$ as well as in the proof
of Lemma~\ref{lemma:lemma6.1} to solve the leaf problem of $(G,T)$.

\begin{lemma}
\label{lemma:lemma5.5}
For any given positive integers $k=O(\log\log m)^{O(1)}$ and $w=O(\log
m)^{O(1)}$, it takes $O(m)$ time to compute a data structure $\lookuptable(k,w)$
such that the following statements hold for any $O(1)$-degree graph
$H$ with at most $k$ nodes whose edge weights are at most $w$.
\begin{enumerate}
\item 
\label{item:lemma5.5:item1}

It takes $O(|H|)$ time to obtain a reference pointer $\textit{ref}(H)$
from $\lookuptable(k,w)$ such that each of the following queries for
any two distinct nodes $u$ and $v$ of $H$ can be answered from
$\textit{ref}(H)$ and $\lookuptable(k,w)$ in $O(1)$ time: (1) the
distance of $u$ and $v$ in $H$, (2) an edge incident to $u$ that
belongs to at least one min-weight path between $u$ and $v$ in $H$,
and (3) the distance of $u$ and $v$ in $H\setminus \{e\}$ for each
edge $e$ of $H$ incident to $u$.

\item 
\label{item:lemma5.5:item2}
It takes $O(|H|)$ time to obtain $\girth(H)$ from $\lookuptable(k,w)$.
\end{enumerate}
\end{lemma}
\begin{proof}
Let $\mathcal{H}$ consist of all graphs of at most $k$ nodes whose
maximum weight is at most $w$.  It takes $O(w)^{O(k^2)}$ time to list
all graphs $H$ in $\mathcal{H}$. It takes $O(k^{O(1)})$ time to
pre-compute the information in Statements~1 and~2 for each graph $H$
in $\mathcal{H}$.  The lemma follows from $$\left(O(\log
m)^{O(1)}\right)^{\left(O(\log\log m)^{O(1)}\right)}\cdot
O\left((\log\log m)^{O(1)}\right)^{O(1)}=O(m).$$
\end{proof}

\begin{lemma}
\label{lemma:lemma5.6}
Let $G$ be an $m$-node $O(1)$-degree $r$-outerplane graph with
$\wmax(G)=O(\log^2 m)$.  Given an $r$-dissection tree $T$ of $G$, the
border problems of $(G,T)$ for all special leaf vertices of $T$ can
be solved in $O(mr)$ time.
\end{lemma}

\begin{proof}
We assume that $T$ does have special leaf vertices, since otherwise
the lemma holds trivially. By the assumption, we know $r\leq
\lceil\log^2 \ell(m)\rceil$.  Let $L$ be a special leaf vertex of
$T$.  Let $G_L=G[L]$.  Let $m_L=|L|$.  By Condition~\ref{condition:cond2} 
of $T$, we know $m_L=\Theta(\ell(m))$.  Let $r_L=r+|\border(L)|$.
Clearly, $G_L$ is an $m_L$-node $O(1)$-degree
$r_L$-outerplane graph with $r_L=O(\log^2 m_L)$.  By
Lemma~\ref{lemma:lemma4.1}, it takes $O(m_L)$ time to obtain an
$r_L$-dissection tree $T'_L$ of $G_L$. Let $T_L$ be obtained from
$T'_L$ by replacing each vertex $S'$ of $T'_L$ by $S'\cup
\border(L)$.
\begin{figure}[t]
\centerline{\scalebox{0.85}{\input{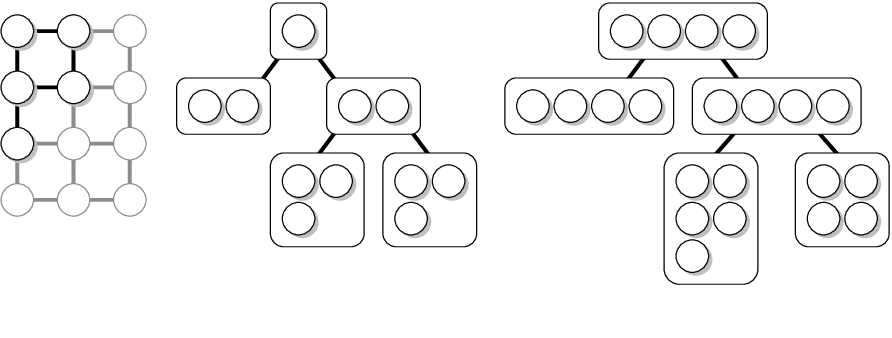}}}
\caption{(a) Graph $G_L=G[L]$ with $L=\{2,3,4,7,8\}$. (b) A dissection
  tree $T'_L$ of $G_L$.  (c) A dissection tree $T_L$ of $G_L$
  obtained from $T'_L$.}
\label{figure:fig7}
\end{figure}
For instance, let $T$ and $G$ be as shown in
Figures~\ref{figure:fig6}(a) and \ref{figure:fig6}(b). 
If $L=\{2,3,4,7,8\}$ is a special leaf vertex of $T$,
then $G_L$ is as shown in Figure~\ref{figure:fig7}(a). We have
$\border(L) = \{2,7,8\}$.  If $T'_L$ is as shown in
Figure~\ref{figure:fig7}(b), then $T_L$ is as shown in
Figure~\ref{figure:fig7}(c).
Clearly, $\border(L)\subseteq \treeroot(T_L)$.  We show that $T_L$
is also an $r_L$-dissection tree of $G_L$.  Since $L$ is a leaf vertex
of $T$, we have $\border(L)\subseteq L$.  Therefore,
Properties~\ref{property:prop1} and~\ref{property:prop2} of $T_L$
follow from Properties~\ref{property:prop1} and~\ref{property:prop2}
of $T'_L$.  Let $\ell_L=\ell(m_L)$.  By
Condition~\ref{condition:cond1} of $T'_L$ and $|\border(L)|=O(r_L)$,
we have $|V(T_L)|=|V(T'_L)|=O(m_L/\ell_L)$ and
\begin{displaymath}
\sum_{\hat{L}\in\leaf(T_L)}|\border(\hat{L})|\leq
|V(T'_L)|\cdot|\border(L)|+\sum_{L'\in\leaf(T'_L)}
|\border(L')|=O(m_L\cdot r_L/\ell_L).
\end{displaymath}
Condition~\ref{condition:cond1} holds for $T_L$.  Adding $\border(L)$
to vertex $S'$ of $T'_L$ increases $|S'|$ and $|\border(S')|$ by no
more than $r_L$, so Conditions~\ref{condition:cond2}
and~\ref{condition:cond3} for $T_L$ follow from
Conditions~\ref{condition:cond2} and~\ref{condition:cond3} for
$T'_L$.  Therefore, $T_L$ is an $r_L$-dissection tree of $G_L$ with
$\border(L)\subseteq \treeroot(T_L)$.  It follows that a solution to
the border problem of $(G_L, T_L)$ for $\treeroot(T_L)$ yields a
solution to the border problem of $(G,T)$ for $L$.

Let $k$ be the maximum $|\hat{L}|$ over all leaf vertices $\hat{L}$ of
$T_L$ and all special leaf vertices $L$ of $T$. We have
$k=\Theta(\ell_L)=O((\log\log m)^{30})$.  By $\wmax(G)=O(\log^2 m)$,
it takes $O(m)$ time to compute a data structure
$\lookuptable(k,\wmax(G))$, as ensured by Lemma~\ref{lemma:lemma5.5}.
By Lemma~\ref{lemma:lemma5.5}(\ref{item:lemma5.5:item1}), it takes
$O(|\hat{L}|+|\border(\hat{L})|^2)=O(\ell_L)$ time to obtain from the
pre-computed data structure $\lookuptable(k,\wmax(G)))$ a solution to
the border problem of $(G_L,T_L)$ for each special leaf vertex
$\hat{L}$ of $T_L$.  By Condition~\ref{condition:cond1} of $T_L$, the
border problems of $(G_L,T_L)$ for all special leaf vertices of $T_L$
can be solved in overall $O(m_L / \ell_L) \cdot O(\ell_L) = O(m_L)$
time.  By Lemma~\ref{lemma:lemma5.4}, it takes $O(m_L\cdot r_L)$ time
to obtain a collection of solutions to the border problems of
$(G_L,T_L)$ for all vertices of $T_L$, including $\treeroot(T_L)$,
which yields a solution to the border problem of $(G,T)$ for the
special leaf vertex $L$ of $T$.  By Condition~\ref{condition:cond1} of
$T$ and $O(m_L\cdot r_L)=O(\ell(m)\cdot(r+|\border(L)|))$, the overall
running time to solve the border problems of $(G,T)$ for all special
leaf vertices of $T$ is $O(\ell(m))\cdot\sum_{L\in\leaf(T)}
O(r+|\border(L)|)=O(mr)$.  The lemma is proved.
\end{proof}

\section{Task 3: solving the leaf problem}
\label{section:task3}

\begin{lemma}
\label{lemma:lemma6.1}
Let $G$ be an $m$-node $O(1)$-degree $r$-outerplane graph satisfying that
$\wmax(G)+r=O(\density(G))$.  Given an $r$-dissection tree $T$
of $G$, the leaf problem of $(G,T)$ can be solved in
$O(m\cdot\density(G))$ time.
\end{lemma}

\begin{proof}
If $\density(G)\geq\log^2\ell(m)$, by Condition~\ref{condition:cond1}
of $T$ and Theorem~\ref{lemma:lemma2.1}, the problem can be solved in
$O(\ell(m)\log^2 \ell(m))\cdot O(m/\ell(m))=O(m\cdot\density(G))$
time.  The rest of the proof assumes
$\wmax(G)+r=O(\density(G))=O(\log^2\ell(m))$.  Let $L$ be a leaf
vertex of $T$. Let $m_L=|L|$.  Let $G_L=G[L]$.  By
Condition~\ref{condition:cond2} of $T$, we have $m_L=\Theta(\ell(m))$.
Therefore, $G_L$ is an $m_L$-node $O(1)$-degree $r$-outerplane graph
with $\wmax(G_L)+r=O(\log^2 m_L)$.  By Lemma~\ref{lemma:lemma4.1}, an
$r$-dissection tree $T_L$ of $G_L$ can be obtained from $G_L$ in
$O(m_L)$ time.  Let $k$ be the maximum $|\hat{L}|$ over all leaf
vertices $\hat{L}$ of $T_L$ and all leaf vertices $L$ of $T$.  We have
$k=\Theta(\ell(m_L))=O((\log\log m)^{30})$.  Let
$\lookuptable(k,\wmax(G))$ be a data structure computable in $O(m)$
time as ensured by Lemma~\ref{lemma:lemma5.5}. By
Lemma~\ref{lemma:lemma5.5}(\ref{item:lemma5.5:item2}),
$\girth(G_L[\hat{L}])$ for any leaf vertex $\hat{L}$ of $T_L$ can be
obtained from $\lookuptable(k,\wmax(G))$ in $O(|\hat{L}|)$ time.  By
Conditions~\ref{condition:cond1} and~\ref{condition:cond2} of $T_L$,
the solution to the leaf problem of $(G_L, T_L)$ can be obtained from
$\lookuptable(k,\wmax(G))$ in $O(m_L/\ell(m_L))\cdot
O(\ell(m_L))=O(m_L)$ time.  By Lemma~\ref{lemma:lemma5.1}, the nonleaf
problem of $(G_L, T_L)$ can be solved in $O(m_L\cdot r)$ time.  By
Conditions~\ref{condition:cond1} and~\ref{condition:cond3} of $T_L$,
we have $\squares(T_L)=O(m_L\cdot r^2\log^2 m_L/\ell(m_L))=O(m_L)$.
By Lemma~\ref{lemma:lemma3.1}, it takes $O(m_L)$ time to compute
$\girth(G_L)$ from the solutions to the leaf and nonleaf problems of
$(G_L, T_L)$.  Therefore, $\girth(G[L])$ can be computed in
$O(m_L\cdot r)=O(\ell(m)\cdot r)$ time.  By
Condition~\ref{condition:cond1} of $T$, it takes $O(m/\ell(m))\cdot
O(\ell(m)\cdot r)=O(m\cdot\density(G))$ time to solve the leaf problem
of $(G,T)$. The lemma is proved.
\end{proof}

It remains to prove the main lemma of the paper, which implies
Theorem~\ref{theorem:theorem1.1}, as already shown
in~\S\ref{subsection:proving-main-theorem}.

\begin{proof}[Proof of Lemma~\ref{lemma:lemma2.6}]
Let $m=|V(G)|$ and $n=|V(\expand(G))|$.  Let $r= \radius(G)$.  That
is, $G$ is an $m$-node $O(1)$-degree $r$-outerplane graph with
$\wmax(G)+r=O(\density(G))=O(\log^2 m)$.  By
Lemma~\ref{lemma:lemma4.1}, an $r$-dissection tree $T$ of $G$ can be
obtained from $G$ in $O(m)$ time.  By Lemma~\ref{lemma:lemma5.1}, the
nonleaf problem of $(G,T)$ can be solved in $O(mr)=O(n)$ time.  By
Lemma~\ref{lemma:lemma6.1}, it takes $O(m\cdot\density(G))=O(n)$ time
to solve the leaf problem of $(G,T)$.  By
Conditions~\ref{condition:cond1} and~\ref{condition:cond3} of $T$, we
have $\squares(T)= O(mr^2\log^2 m/\ell(m))=O(m)$.  The lemma follows
from Lemma~\ref{lemma:lemma3.1}.
\end{proof}

\section{Concluding remarks}
\label{section:conclude}
We give the first linear-time algorithm for computing the girth of any
undirected unweighted planar graph. Our algorithm can be modified into
one that finds a simple min-weight cycle. Specifically, when we solve
each girth problem or each distance problem in our algorithm, we
additionally let the algorithm output a node on a corresponding
min-weight cycle or min-weight path. As a result, our algorithm not
only computes the girth of the input graph, but also outputs a node
$u$ on a min-weight cycle of the input graph.  We can then use the
breadth-first search algorithm of Itai and
Rodeh~\cite{ItaiR78} to output a min-weight cycle containing
$u$ in linear time.  

The $O(n \log n)$-time algorithm of Weimann and
Yuster~\cite{WeimannY10} works on $O(1)$-genus graphs. It would be of
interest to see if our algorithm can be extended to work for
$O(1)$-genus graphs by, e.g., extending our black-box tools (the
decomposition tree of Goodrich~\cite{Goodrich95} and the distance
oracle of Klein~\cite{Klein05}) to work for $O(1)$-genus graphs.

\section*{Acknowledgment}
We thank the anonymous reviewers for their helpful comments.
We also thank Hsueh-Yi Chen, Chia-Ching Lin, and Cheng-Hsun Weng for
discussion at an early stage of this research.

\bibliographystyle{abbrv}
\bibliography{girth}
\end{document}